\documentclass[final,onefignum,onetabnum]{siuro210301}

\usepackage{lipsum}
\usepackage{amsfonts}
\usepackage{graphicx}
\usepackage{epstopdf}
\usepackage{algorithmic}
\ifpdf
  \DeclareGraphicsExtensions{.eps,.pdf,.png,.jpg}
\else
  \DeclareGraphicsExtensions{.eps}
\fi

\usepackage{enumitem}
\setlist[enumerate]{leftmargin=.5in}
\setlist[itemize]{leftmargin=.5in}

\newsiamremark{remark}{Remark}
\newsiamremark{hypothesis}{Hypothesis}
\crefname{hypothesis}{Hypothesis}{Hypotheses}
\newsiamthm{claim}{Claim}

\headers{Testing for the Important Components of Posterior Predictive Variance}{D. Dustin, and  B. Clarke}

\title{Testing for the Important Components of Posterior Predictive Variance}

\author{Dean Dustin \thanks{Department of Statistics, University of Nebraska, Lincoln,  
  (\email{ddustin8@huskers.unl.edu}).}
\and Bertrand Clarke \thanks{Department of Statistics, University of Nebraska, Lincoln, 
  (\email{bclarke3@unl.edu}).}
}

\usepackage{amsopn}

\ifpdf
\hypersetup{
  pdftitle={Testing for the Important Components of Posterior Predictive Variance},
  pdfauthor={D. Dustin,  and B. Clarke}
}
\fi

\begin{document}

\maketitle

\begin{abstract}
We give a decomposition of the posterior predictive variance using the
law of total variance and conditioning on a finite dimensional discrete random variable.
This random variable summarizes various features of modeling that are used
to form the prediction for a future outcome.
Then,
we test which terms in this decomposition are small enough to ignore.  This allows us
identify which of the discrete random variables are most important
to prediction intervals.  The terms in the decomposition admit interpretations 
based on conditional means and variances and are analogous to the terms in a
Cochran's theorem decomposition of squared error often used in
analysis of variance.  Thus, the modeling features are treated as factors
in  completely randomized design.   In cases where there are multiple decompositions
we suggest choosing the one that that gives the best predictive coverage with the
smallest variance.  

\textbf{Keywords} --- prediction intervals, posterior predictive variance,  law of total variance, 
 Bayes model averaging,  stacking, ANOVA, bootstrap testing,Cochran's theorem.%
\end{abstract}%

\section{Introduction}
\label{intro}

At the risk of oversimplification, it is usually the case that
more complex data leads to more complex models and more complex models
in turn lead to greater
demands for validation.    Moreover, the ultimate form of validation is predictive:
Models that do not achieve good prediction are discredited.    It is a slipperier question
when two models achieve similar predictive performance in finite samples although
in such cases we may sometimes rely on asymptotics, robustness, or other
properties to help us decide which model
is more appropriate -- assuming they are incompatible.

In the Bayesian context, the most common predictor is the posterior predictive distribution
that has density $p(Y_{n+1} \vert {\cal{D}}_n)$, where ${\cal{D}}_n$ is the data available
before the next response $Y_{n+1}$ is revealed.  The posterior predictive is optimal under a
relative entropy criterion.  Moreover, often, the squared error loss is invoked to justify
using the posterior predictive mean $E(Y_{n+1} \vert {\cal{D}}_n)$ as a predictor
for $Y_{n+1}$.  In this case, predictive intervals (PI's) are derived from the distribution of
\begin{eqnarray}
\frac{Y_{n+1} - E(Y_{n+1} \vert {\cal{D}}_n)}
{\sqrt{\hbox{Var}(Y_{n+1} \vert {\cal{D}}_n) } }.
\label{CLTform}
\end{eqnarray}
That is, we find $c_\alpha$, $\alpha > 0$, so that
\begin{eqnarray}
 1-\alpha \leq 
P\left(  \left|  \frac{Y_{n+1} - E(Y_{n+1} \vert  {\cal{D}}_n) }
{\sqrt{\hbox{Var}(Y_{n+1} \vert  {\cal{D}}_n) }} \right|
< c_\alpha \left| \right. { \cal{D}}_n \right)
\nonumber
\end{eqnarray}
where the probability $P$ applies to $Y_{n+1}$ and hence obtain 
\begin{eqnarray}
PI(\alpha) = 
E(Y_{n+1} \vert  {\cal{D}}_n ) \pm c_\alpha \sqrt{\hbox{Var}(Y_{n+1} \vert {\cal{D}}_n )}.
\label{PIalpha}
\end{eqnarray}
From \eqref{PIalpha}, it is easy to see that $E(Y_{n+1} \vert  {\cal{D}}_n ) $ controls the location
of the PI while $\hbox{Var}(Y_{n+1} \vert {\cal{D}}_n )$ controls its width.
When $p(Y_{n+1} \vert {\cal{D}}_n)$ is symmetric, unimodal and based on well-behaved data,
e.g., many classes of independent data,
$c_\alpha$ is a quantile of the posterior predictive distribution.  This can be
extended to more general distributional shapes.

It is seen that \eqref{CLTform} is much like a $t$-statistic.  Indeed, 
\begin{align*}
Var(Y_{n+1} - E(Y_{n+1} \vert  {\cal{D}}_n) \vert {\cal{D}}_n) & = Var(Y_{n+1} \vert  {\cal{D}}_n) + Var(E(Y_{n+1} \vert  {\cal{D}}_n) \vert  {\cal{D}}_n)\\
&= Var(Y_{n+1} \vert {\cal{D}}_n) + E(E(Y_{n+1} \vert  {\cal{D}}_n)^2 \vert  {\cal{D}}_n) - E(E(Y_{n+1} \vert  {\cal{D}}_n) \vert  {\cal{D}}_n)^2\\
&= Var(Y_{n+1} \vert  {\cal{D}}_n) + E(Y_{n+1}\vert  {\cal{D}}_n)^2  - E(E(Y_{n+1} \vert  {\cal{D}}_n) \vert  {\cal{D}}_n)E(E(Y_{n+1} \vert  {\cal{D}}_n) \vert  {\cal{D}}_n)\\
&= Var(Y_{n+1} \vert  {\cal{D}}_n) + E(Y_{n+1}\vert  {\cal{D}}_n)^2  -  E(Y_{n+1}\vert  {\cal{D}}_n)^2  \\
&= Var(Y_{n+1} \vert  {\cal{D}}_n).
\end{align*}
So, \eqref{CLTform} is of the form of a location divided by its root variance or standard error.

The goal of this paper is to present an additive
decomposition for $\hbox{Var}(Y_{n+1} \vert {\cal{D}}_n )$
that has three key properties:  i) The terms are individually interpretable as a sort of
variability intrinsic to $Y_{n+1}$; ii) Each term can be tested to see if it is small enough relative
to the other terms that it can be neglected, and iii) Taken together the decomposition of
$\hbox{Var}(Y_{n+1} \vert {\cal{D}}_n )$ is analogous to Cochran's theorem including
allowing flexibility as to how many terms are included.
The implication of this analysis is that the components of posterior predictive variance
can be examined to determine what they say about the various ingredients used to
formulate the model.
That is, we may consider a variety of modeling schemes with different components and test 
to see which is most appropriate and then within the most appropriate modeling
scheme test which terms in the posterior predcitive variance should 
be retained -- or discarded.

More pragmatically, we present methodology for expanding or reducing a model list. 
The general methodology rests on the use of hypothesis testing 
to determine which of the choices we make to form predictors affect
the predictive distribution most. Also, we consider coverage and width of PI's
 from competing posterior predictive distributions.  Our overall goal is to form the smallest prediction
intervals possible that have close to the nominal coverage.   
This is seen in the example in Sec. \ref{example}.

Our analysis rests on applying the law of iterated variances to future outcomes.  Let
$V$ be a random variable and write the predictive variance decomposition
\begin{eqnarray}
\hbox{Var}(Y_{n+1} \vert {\cal{D}}_n) = E ( \hbox{Var}(Y_{n+1} \vert V, {\cal{D}}_n) )
+ 
 \hbox{Var} (E (Y_{n+1} \vert V, {\cal{D}}_n)).
\label{lawtotalvariance}
\end{eqnarray}
In our examples here, $V$ will typically be discrete although continuous $V$'s
satify \eqref{lawtotalvariance} as well.
The first term on the right 
can be interpreted as the average location of the variance taking into account
the variability of $V$.  The second term is the variability contributed by $V$ to
the location of the predictive distribution.   If the second term is small, then we
know that $(Y_{n+1} \vert {\cal{D}}_n)$ is not affected much by the variability of $V$
so it may make sense to ignore this term.  On the other hand, if the first term
is small,  then the contribution of $V$ to the variance of $(Y_{n+1} \vert {\cal{D}}_n)$
may be ignored.  The distinction between these two terms is whether $V$ affects
the variability in location or the variability in variance.  

Loosely,  the values $V$ assumes represent some feature of the modeling strategy
for the sequence of random variables in ${\cal{D}}_n = \{x_1, y_1; \ldots ; x_n, y_n\}$
where the $x_i$'s are $p$-dimensional explanatory variables giving response $y_i$
under some error struture.   For instance, as seen in Sec. \ref{example}, $V$ may represent
the choice of a shrinkage method in penalized linear regression.   In other
examples here, $V$ may represent a link function in generalized
linear models,  a nonlinear regression technique., or a selection of variables.   

We can take $V = V_K= (V_1, \ldots , V_k, \ldots, V_K)$ and apply \eqref{lawtotalvariance} iteratively to itself, 
generating a new term at each iteration.  This gives us $K+1$ terms that can be
interpreted in terms of means and variances.   Thus,  we must choose
a $K$ and we can regard each $V_k$ as an aspect of a modeling strategy.  
For instance, $V_1$ may be a `scenario' and $V_2$ may be a `model' in the sense of
\cite{Draper:1995}, a parallel we develop in Sec. \ref{challenger}.  
If we write
$\hbox{Var}_{V_K}(Y_{n+1} \vert {\cal{D}}_n)$ to mean the posterior predictive variance
using a specific choice of $V_K$, it is easy to see, in general, that for another 
choice, say, $V^\prime_{K^\prime}$,
we will usually find
$\hbox{Var}_{V_K}(Y_{n+1} \vert {\cal{D}}_n) 
\neq \hbox{Var}_{V^\prime_{K^\prime}}(Y_{n+1} \vert {\cal{D}}_n)$.   On the other hand, the relative sizes of terms in 
decompositions of the form \eqref{lawtotalvariance} depend delicately on the 
choice of $K$ and $V_K$.   Fortunately,
in practice, we usually only have one $V_K$ that we most want to consider,  but the
order may matter and it is partially a matter of statistical judgement how big
$V_K$ should be and what components it should have.

One choice of $V_K$, with $K=2$, that can be used in general to expand a model list
to capture more possibilities and then winnow down to the most successful of them
is the following.    Consider trying to assess the importance of sets of variables in a 
predictive modeling situation. Suppose we have a list of models that we are considering, 
${\cal{M}} = \{m_1, \ldots, m_q \}$ and we have a set of explanatory variables 
${\cal{X}} = \{ X_1, \ldots, X_p\}$. Write ${\cal{P(X)}} = \{ \{X\}_1,\ldots , \{X\}_{2^p} \}$ for
the power set of ${\cal{X}}$. Now we can consider each model with each subset of 
explanatory variables as inputs to the modeling.   Here, $V_1$ corresponds to
the uncertainty in the predictive problem due to the models and $V_2$
corresponds to the variables we use in the models.   We give an example
of this in Subsec. \ref{disaster}.

Using a Bayes model average we write the posterior predictive density as 
\begin{eqnarray}
p(Y_{n+1} \vert {\cal{D}}_n) 
= \sum^q_{i=1} p( m_i \vert {\cal{D}}_n) \sum^{2^p}_{j=1} 
p(\{X\}_j \vert {\cal{D}}_n, m_i ) p(Y_{n+1} \vert {\cal{D}}_n, \{X\}_j, m_i),
\label{postpredBMA}
\end{eqnarray}
generically denoting prior densities as $p$.
Now, we can the calculate posterior probability for each 
set of explanatory variables from
\begin{eqnarray}
 p(\{X\}_j\vert {\cal{D}}_n) = \sum^q_{i=1} p( m_i \vert {\cal{D}}_n) p(\{X\}_j \vert {\cal{D}}_n, m_i ).
 \label{postprobvarset}
\nonumber
\end{eqnarray}
This posterior probability is a measure of ``variable set importance''.  A similar expression gives 
a measure of importance for an individual model.   
 

Once $K$ and $V_K$ have been chosen, the decomposition based on $V_K$ can
be generated and examined for which terms are important.  We do this using a bootstrap
testing procedure.  We regard our testing procedure as an approximation to the
tests that emerge from a Cochran's theorem decomposition that are known to
be $F$-tests.  The reason is that, at least superficially, our general posterior
predictive decomposition
resembles the Cochran's theorem decomposition of the squared error into a
sum of quadratic forms; see Subsec. \ref{analog}.
In fact, our tests resemble ratios of $\chi$-squared distributions
but we cannot ensure the independence or determine the degrees of freedom explicitly.
We resort to a bootstrap procedure since the terms we want to test for
proximity to zero are latent quantities, i.e.,  they do not directly depend on the data, 
and hence do not have an accessible likelihood.  The overall procedure clearly involves
several steps; an example demonstrating our basic 
methodology is given in Sec.\ref{example}.



Another point bears comment:  We have written our methodology in terms of
conditioning and the Bayesian approach.  We think this is the `right' way to apply the
information in the data to an inferential problem.    Nevertheless, 
many authors do not regard Bayes model averaging
as the only or main way to express model uncertainty.   For instance, stacking coefficients,  
see \cite{Wolpert:1992}, are often used for model averaging and can also be
regarded as summaries of model uncertainty.   Indeed, in some cases, stacking coefficients,
being based on a cross-validation optimization,  may be easier to work with.
Moreover,  \cite{Zhang:Liu:2022} show that when the true model is on the model list,
its stacking weight is asymptotically one and otherwise the stacking average converges
to the predictively optimal model on the list.
In addition,  \cite{Yao_etal_2018} and \cite{Clarke:2003}
argues that for predictive purposes stacking distributions and means, respectively,
often outperforms Bayes model averaging.
Consequently, in our work below, we sometimes give
the stacking analog to \eqref{lawtotalvariance} to demonstrate the generality of
our approach.

Continuing the example of \eqref{postpredBMA},
suppose the collection of models ${\cal{M}}$ are not easily implementable in the Bayesian setting. Then we can use stacking instead and define a similar predictive distribution by 
\begin{eqnarray}
p(Y_{n+1} ) 
= \sum^q_{i=1} w( m_i ) \sum^{2^p}_{j=1} w(\{X\}_j \vert m_i ) p(Y_{n+1} \vert \{X\}_j, m_i),
\label{postpredstack}
\nonumber
\end{eqnarray}
where the dependence of the summands on ${\cal{D}}_n$ is suppressed.
Now,  the stacking weights for each set of explanatory variables are
 $$
 w(\{X\}_j ) = \sum^q_{i=1} w( m_i ) w(\{X\}_j \vert m_i ).
 $$
Computing these weights is a quadratic programming problem that can be done easily for a reasonable number of explanatory variables.   

The structure of this paper is as follows.  We begin in Sec.\ref{example} with an
example of how our methodology can be used to determine which shrinkage method
within a finite collection of shrinkage methods is best in the sense of minimizing
the posterior predictive variance.    The justification for our method is only briefly mentioned
since the focus is on implementation.    Sec. \ref{derive} presents our full method
with justifications.  There are subsections to explain the variance decomposition
in terms of quadratic forms and the testing procedure for terms in a variance decomposition.
Secs. \ref{challenger} and
\ref{superconduct} contain two further examples.  The first shows how our work goes
beyond \cite{Draper:1995} and the second shows a complete real data example
where a method such as the one we propose is more useful for uncertainty quantification
than other
conventional direct modeling methods, at least predictively.
We conclude with a discussion of our overall contribution in Sec. \ref{discussion}.

\section{A Numerical Example}
\label{example}

An example will make the point regarding the importance of 
the last term in \eqref{lawtotalvariance}.
There has been much discussion about when different shrinkage methods are
appropriate, see \cite{Wang:etal:2020} for instance.   
The consensus from simulations and applications
seems to be that for easy, general use LASSO or Elastic Net (EN,a generalization of LASSO)
are usually best when there is enough sparsity in the data and multicollinearity is not
a problem.  Otherwise,  when sparsity is low or multicollinearity is a problem 
ridge regression is preferable.
Our techniques provide a more formal basis for this intuitive summary of examples.

Let us compare five penalized methods, namely LASSO,
Ridge Regression (RR), Adaptive LASSO (ALASSO), EN, and Adaptive EN (AEN)
applied to a linear model 
\begin{eqnarray}
Y_i = X_i^T \beta + \epsilon_i
\nonumber
\end{eqnarray}
for $i=1, \ldots , n$ where $X_i$ is a vector of explanatory variables with 
$\dim(X_i) = \dim(\beta) = p < n$
and $\epsilon_i \sim N(0, 1)$ IID
and write $m_1, \ldots, m_5$ to mean the five penalty functions.
Write $V$ to be the discrete random variable assuming values over the five methods
i.e., over the $m_j$'s.   

Let us now apply the two term variance decomposition in \eqref{lawtotalvariance}
using $V$.  We suspect that the second term 
on the right is small relative 
to the left hand side because we think the models from the five methods will be very similar,
i.e., they will have similar locations even if their variances are not identical.   That is,
we suspect that a hypothesis test of
$$
H_0:  E\left( \frac{Var_{V}E(Y_{n+1}\vert {\cal{D}}_n, V)}{Var(Y_{n+1} \vert  {\cal{D}}_n)}\right) \geq 0.05
$$ 
versus 
$$
H_1:  E\left( \frac{Var_{V}E(Y_{n+1}\vert {\cal{D}}_n, V)}{Var(Y_{n+1} \vert  {\cal{D}}_n)}\right) < 0.05
$$
will end up rejecting the null, meaning we can drop the second term in \eqref{lawtotalvariance}
and the .05 level.

To investigate the behavior of the terms in the predictive variance decomposition 
we generate data as follows.    Let $n=50$ and $p=100$ and take 95 of the $\beta_j$ 
coefficients to be zero and five to be generated independently from a $N(5, (1.5)^2)$.  
We will see below, Subsec. \ref{testing}, that this test can be performed by bootstrapping
the argument of the expectation in the null hypothesis.  In fact, for normal error, the
distributions of the numerator and the denominator are, approximately, convex combinations 
of $\chi^2$ distributions.  So their ratio is expected to behave like an $F$ distribution.
The convex combinations can be exactly defined but are generally inaccessible
numerically.  Consequently, our bootstrap-based testing procedure is a simplified nonparamertic
approximation to the standard normal theory.   

To set up our analysis of the simulated data, we used the first 49 data points to form
predictive distributions for each of the five methods as well as for the stacking
average (based on five-fold cross-validation ) of the five methods.  That is,
we found the stacking weights $\hat{w}_1, \ldots , \hat{w}_5$
as well as the $\beta$ coefficients and the decay parameters for each of the five
methods.    For the stacking weights we imposed both the positivity and sum-to-one
constraints.   We use stacking rather than Bayes model averaging in
this example because stacking weights are predictive by construction and
we used {\sf glmnet}, a frequentist implementation, for our computations.  

We write the stacking model average as
\begin{eqnarray}
\sum_{j=1}^5 \hat{w}_j({\cal{D}}_{49}) p( Y_{50} \vert X_{50}, m_j)({\cal{D}}_{49})
\label{stackmix}
\end{eqnarray}
where we have indicated the dependence of the $\hat{\beta}_j$'s in the model
$p$ by writing ${\cal{D}}_{49}$ in parentheses on the right.    We also set
\begin{eqnarray}
p( Y_{50} \vert X_{50}, m_j)({\cal{D}}_{49}) = N( X_{50} \hat{\beta}_{m_j}, 
\hat{\sigma}^2_{m_j} + \widehat{Var}(X_{50}\hat{\beta}_{m_j}) )
\label{normalpredictive}
\end{eqnarray}
where the estimation of the decay parameters
$\lambda_j$ is suppressed in the $m_j$'s and $\hat{\sigma}^2_{m_j}$ 
is the standard OLS estimator of $\sigma^2$ using only the variables selected by $m_j$
-- except for RR where we use the $\hat{\sigma}$ from EN since it is a combination
of the $L^1$ and $L^2$ penalties.  We justify this by citing \cite{Zhao:etal:2021}
who showed that this procedure is consistent for LASSO.  We also observe that
the proof can be extended to EN and, we think, to any shrinkage method
with the oracle property (e.g., AEN and ALASSO).
To find $\widehat{Var}(X_{50}\hat{\beta}_{m_j})$ 
we use the bootstrapped variance estimator from the
{\sf boot} package in R.  

Now, we draw another $100,000$ data points from each of the five models.
Then we sample $n_j = 100,000\hat{w}_j$ from each model 
$p( Y_{50} \vert X_{50}, m_j)({\cal{D}}_{49})$, for $j=1, \ldots, 5$.
This gives us 100,000 data points from the stacking mixture \eqref{stackmix}.  We use these
data points to assess coverage of the PI's from the five shrinkage methods
and their stacking average.  The PI for stacking is of the form
$PI_{stack}(.05) = [ q_{.025}, q_{.975}]$ where the $q$'s are the quantiles
from \eqref{stackmix}.  Similarly, we have $PI_{m_j}(.05) = [q_{.025}^{m_j}, q_{.975}^{m_j}]$.
This gives us 6 PI's.  

To estimate the empirical coverage of the six PI's, we use the bootstrap again 
now on the entire procedure up to this point.    We choose $B=1000$.
Letting $j=1, \ldots , 6$ index the predictive distributions --- $j=6$ corresponds to the
stacking average --- the result is
\begin{eqnarray}
\widehat{\sf Coverage}_j = \frac{1}{B}\sum_{b=1}^{B} \chi_{\{Y_{b} \in PI_{j,b}\}}.
\nonumber
\end{eqnarray}
We also have the bootstrapped variance from the $j$-th predictive
distribution from the RHS of \eqref{normalpredictive}.  This
procedure bootstraps the three terms in \eqref{lawtotalvariance}.
The details on enforcing the null hypothesis are in Subsec. \ref{testing}.
Essentially, we get a bootstrapped $p$-value, commonly called the
achieved significance level (ASL), and reject when the ASL is too small.
Our computed results are summarized in \Cref{simexample}.

\begin{table}[H]
\caption{\textbf{Stacking shrinkage methods}:  This table gives the stacking weights,
the variances of the predictive distributions, and the coverage of the PI's
for five shrinkage methods and their stacking average.}
\label{simexample}
\centering
\begin{tabular}{rrrrrrr}
  \hline
 & STK avg &  LASSO & RR & ALASSO & EN & AEN \\ 
  \hline
Stacking weights &  & 0.74 & 0.00 & 0.00 & 0.25 & 0.00  \\ 
Pred. Variance & 2.97 & 1.02 & 6.71 & 0.99 & 6.73 & 6.70 \\ 
  Coverage & 0.97 & 0.98 & 0.43 & 0.12 & 0.94 & 0.25 \\ 
   \hline
\end{tabular}
\end{table}

We see in \Cref{simexample} that only LASSO and EN have positive stacking weights. 
LASSO achieves greater than the nominal 95\% coverage while EN is slightly less at 94\% 
despite having a much larger predictive variance than LASSO.
The stacked predictive distribution has an estimated variance of $2.97$ and decomposes as 
$$
\widehat{Var}(Y_{50}) = 
\hat{E}_{V}(\widehat{Var}(Y_{50}\vert V)) + \widehat{Var}_{V}\hat{E}(Y_{50}\vert V) = 2.39 + 0.58.
$$
Hence we see the ratio of the between-models variance to total variance is 
$$
\frac{\widehat{Var}_{V}\hat{E}(Y_{50}\vert V)}{\widehat{Var}(Y_{50}) } = \frac{0.58}{2.97} = 0.195.
$$
Informally, this suggests that there is too much between-models variance to ignore 
when making predictions.

More formally, using our test, we obtain an $\widehat{ASL} =0.99$ meaning
we cannot reject the null.  This leads us to conclude
that the second term on the LHS of \eqref{lawtotalvariance} contributes more
than 5\% of the total predictive variance.  
Consequently,  we should account for model uncertainty when making predictions.

Despite both LASSO and EN having good coverage,  the small size of $n$ relative to $p$
leads us to ask what level of between-models variance would lead 
to rejection.?  We observe that
if we change the RHS of $H_0 $ and $H_1 $ to $0.09$ instead of $0.05$, our test 
gives an $\widehat{ASL} = .0095$.  Hence, we would conclude 
that 9\% is the smallest percentage at which we could ignore the contribution
of the between-models variance to the overall variance.

To conclude this example, observe that since we want the correct nominal 
predictive coverage with the smallest $K$ and $V$,
we can look at \Cref{simexample} and reason as follows.
LASSO has smaller or equivalent variance to the other methods
and at least the desired coverage.   We can rule out ALASSO on the basis of
its poor coverage and zero stacking weight.
Thus, if we choose, say, 10\% (or any number bigger than 9\%)
as our threshold,
we are led to use PI's from LASSO.  That is, $V$ reduces to a single level.
We provide more discussion on choosing $K$ and $V$ in Sec. \ref{discussion}.

\section{Decomposing the Posterior Predictive Variance}
\label{derive}

In this section we give our variance decomposition in full generality,
indicate how to choose amongst candidate variance decompositions,
and explain our testing procedure for the terms in a given
variance decomposition.  We will see that our decomposition
of the posterior predictive variance is analogous to the Cochran's theorem decomposition
of the squared error into quadratic forms
used in analysis of variance.  The analogy is limited by the fact
that our terms are Bayesian and only approximately $\chi$-squared.

\subsection{The Effect of the Model List on Overall Variance}
\label{setting}

We can enlarge model list simply by including more plausible models.
However, this may lead to problems such as dilution; see \cite{George:2010}.
So, we want to assess the effect of a model list on the variance
of predictions.    Consider a model list ${\cal{M}}$ and suppose
we don't believe it adequately captures the uncertainty (including mis-specification)
of the the predictive problem. 
We can expand the list by including other competing 
models and this can be done by adding more models to it or by embedding the
models on the list in various `scenarios' as is done in \cite{Draper:1995}. 
Once a new model 
list ${\cal{M}'}$ is constructed, if it contains new models with positive posterior probability, 
the posterior predictive distribution $p(Y_{n+1} \vert {\cal{D}}_n)$ resulting from
${\cal{M}'}$ will be different than $p(Y_{n+1} \vert {\cal{D}}_n)$ from using  
${\cal{M}}$.  More formally, if ${\cal{M}}$ is a model list. Then, the ${\cal{M}}$-dependent 
predictive distribution is
 $$
 p(Y_{n+1} \vert {\cal{D}}_n)=p(Y_{n+1} \vert {\cal{D}}_n)({\cal{M}}).
 $$
In our variance decomposition below, we include the dependence on the model
list by $V$.
Clearly, the typical case is
$Var(Y_{n+1} \vert {\cal{D}}_n)( {\cal{M}'} ) \neq Var(Y_{n+1} \vert {\cal{D}}_n)({\cal{M}} )$,
so the posterior predictive variance depends on the model list i.e., on the choice of $K$
and $V_K$.

\subsubsection{Posterior Predictive Variance Decomposition ``P-ANOVA''}
\label{PANOVA}

To quantify the uncertainty of the subjective choices we must make,
recall $V= (V_{1}, \ldots , V_{K})$, where $V_k$ 
represents the values of the $k$-th potential choice that must
be made to specify a predictor. 
Analogous to terminology
in ANOVA, we call $V_k$ a \textit{factor} in the prediction scheme, and 
we define the levels of $V_k$ to be $v_{k1}, \ldots , v_{km_k}$.  That is,  $v_{k\ell}$ is a 
specific value $V_k$ may assume. Thus, $V$ is discrete and has probability mass 
function $W(v)=W(V_1=v_1 \ldots, V_K = v_K)$.  The $V_k$'s are not in general
independent and $W$ corresponds to a prior on $V$. 
Define our chosen model list by
 \begin{equation}
\label{pred_scheme}
{\mathcal{V}}^{K} =  \{v_{11} ,\ldots , v_{1m_1}\} \cup \ldots \cup \{v_{K1} ,\ldots , v_{Km_K}\}.
\nonumber
\end{equation}
There are $m_1 \times \cdots \times m_K$ distinct models in ${\cal{V}}^{K}$ 
and they may or may not have a hierarchical structure.  

Our first result gives a decomposition of the posterior predictive variance by
conditioning on $V$.  

\begin{proposition}
\label{General_pred_variance_prop}
(BMA Variance)  
We have the following two expressions for the posterior predictive variance.\\
{\it Clause (i):}
For $K=1$ we have \eqref{lawtotalvariance} and for
$K \geq 2$, the posterior predictive variance of $Y_{n+1}$ as function of 
the $K$ factors defining the predictive scheme is 
\begin{align}
\label{Conditional_Var_sum}
\nonumber 
Var(Y_{n+1} \vert  {\cal{D}}_n)(  {\mathcal{V}}^{K}) 
& = E_{(V_1,\ldots, V_k)} Var(Y_{n+1} \vert  {\cal{D}}_n, V_1, \ldots, V_K) \\ \nonumber 
& + \sum_{k=2}^{K} E_{(V_{1}, \ldots, V_{k-1} )} Var_{V_k}E(Y_{n+1} \vert  {\cal{D}}_n, V_1, \ldots, V_k) \\
& + Var_{V_1}E(Y_{n+1} \vert  {\cal{D}}_n, V_1).
\end{align}
\\
{\it Clause (ii):}
For any $K$, the posterior predictive variance 
$Var(Y_{n+1} \vert  {\cal{D}}_n)(  {\mathcal{V}}^{K})$ can be condensed into a two term decomposition:
\begin{align}
\label{condensed_var}
\nonumber Var(Y_{n+1} \vert  {\cal{D}}_n)( {\mathcal{V}}^{K}) 
& = E_{(V_1,\ldots, V_K)} Var(Y_{n+1} \vert {\cal{D}}_n, V_1, \ldots, V_K) \\ 
& + Var_{(V_1,\ldots, V_K)}E(Y_{n+1} \vert  {\cal{D}}_n, V_1, \ldots, V_K).
\end{align}
\end{proposition}
\begin{proof}
{\it Clause i):}  The proof is a straightforward iterated application of the law of total variance,.

{\it Clause ii):}  This follows from the law of total variance simply treating $V$ as a vector
rather than as the string of its components. $\square$
\end{proof}

We summarize the decomposition in (\ref{Conditional_Var_sum}) using what we call 
``P-ANOVA'', or \textit{predictive analysis of variance}. In \Cref{vartable}, each 
row corresponds to a different source of variability associated with the factors in $V$.

\begin{table}[h]
\caption{\textbf{Sources of Posterior Predictive Variation for $K \geq 3$}.  
We have listed the generic terms in our decomposition of the posterior predictive variance
together with their interpretations.  Following the conventions of ANOVA,  we have also listed
the source of the variability.   All terms are conditional  on ${\cal{D}}_n$.}
\label{vartable}

\centering
\begin{tabular}[t]{|c|c|c|}
\hline Source & Interpretation & Variance \\
\hline
$V_{1} $& Between $V_{1}$ variance & $Var_{V_{1}} E(Y_{n+1} \vert  {\cal{D}}_n, V_{1})$ \\
&&\\
$V_{2}$   & Between $V_{2}$ within $V_{1}$  & $E_{V_{1}}Var_{V_{2}}E(Y_{n+1} \vert  {\cal{D}}_n, V_{1}, V_{2})$ \\
&&\\
\vdots &\vdots& \vdots \\
&&\\
$V_{K}$  &  Between $V_{K}$ within $V_{K-1} \ldots V_{1}  $&$E_{V_{1}}\ldots E_{V_{K-1}}Var_{V_{K}}E(Y_{n+1} \vert  {\cal{D}}_n, V_1 ,V_{2}, \ldots, V_{K})$ \\
&&\\
Predictions  &  within $V_{1}  \ldots V_{k}  $&$E_{V_{1}}\ldots E_{V_{K}}Var(Y_{n+1} \vert  {\cal{D}}_n, V_1 ,V_{2}, \ldots, V_{K})$ \\
\hline
&&\\
Total  &  Posterior predictive variance & $Var(Y_{n+1} \vert  {\cal{D}}_n)$\\
\hline
\end{tabular}
\end{table}%

If we wish to use a frequentist model averaging procedure, rather than a Bayesian, 
we get a similar decomposition, but the variance is a general function of the 
data rather than conditional on the data.   One such result 
is the following.

\begin{proposition} 
\label{General_pred_variance_prop_stacking}
(Stacking Variance)
We have the following two expressions for the stacking predictive variance.\\
{\it Clause (i):}
For $K=1$, the stacking predictive variance for $Y_{n+1}$ is
$$
\hbox{Var}(Y_{n+1})( {\cal{D}}_n, {\mathcal{V}}^{K}) = 
E_{V_1}(\hbox{Var}(Y_{n+1}\vert V_1))( {\cal{D}}_n) + \hbox{Var}_{V_1}E(Y_{n+1}\vert V_1)( {\cal{D}}_n)
$$
and for $K\geq 2$, the stacking predictive variance for $Y_{n+1}$ as function of 
the $K$ factors defining our predictive scheme is given by
\begin{align}
\label{Conditional_Var_sum_stack}
\nonumber 
Var(Y_{n+1})( {\cal{D}}_n,  {\mathcal{V}}^{K}
) 
& = E_{(V_1,\ldots, V_K)} Var(Y_{n+1} \vert  V_1, \ldots, V_K) ({\cal{D}}_n) \\ \nonumber 
& + \sum_{k=2}^{K} E_{(V_{1}, \ldots, V_{k-1} )} Var_{V_k}E(Y_{n+1} \vert  V_1, \ldots, V_k) ({\cal{D}}_n) \\
& + Var_{V_1}E(Y_{n+1} \vert V_1)( {\cal{D}}_n),
\end{align}
where the distribution of $V = (V_1,\ldots, V_K)$ is defined by the stacking weights.
\\
{\it Clause (ii):}
For any $K$, the stacking predictive variance 
$Var(Y_{n+1})(  {\cal{D}}_n, {\mathcal{V}}^{K})$ can be condensed into a two term decomposition:
\begin{align}
\label{condensed_var_stack}
\nonumber Var(Y_{n+1} )(  {\cal{D}}_n, {\mathcal{V}}^{k}) 
& = E_{(V_1,\ldots, V_K)} Var(Y_{n+1} \vert  V_1, \ldots, V_K)({\cal{D}}_n) \\ 
& + Var_{(V_1,\ldots, V_K)}E(Y_{n+1} \vert  V_1, \ldots, V_K)({\cal{D}}_n).
\end{align}
 \end{proposition}



Using stacking -- or any other model averaging procedure -- in place of BMA leads to
a table analogous to \Cref{vartable}. 

\subsection{Analogy to Cochran's Theorem}
\label{analog}

Cochran's theorem is used in standard ANOVA problems to 
identify hypothesis tests that determine whether a factor or its levels
should be dropped as having little effect on the observed variability.
One version of this central result is given in \cite{Rao:1973}. Informally,
the theorem states that,  under various regularity conditions,  the corrected sum of
squares from an ANOVA problem can be written as a sum of independent
quadratic forms each of which is distributed as a $\chi^2$ random variable
with a degrees of freedom specified by the statement of the problem.
Equivalently, the sum of squares ``$Y^TY$'' can be written as a sum of
scaled $\chi^2_1$ random variables, where the 
scaling constants are eigenvalues from the corresponding quadratic form. 
Here we a present a predictive analog of Cochran's Theorem that can be 
used to determine if the means within a $V_k$'s with positive posterior
probabilities (or stacking weights) are different enough that they
contribute substantially to the posterior (or stacking) predictive variance.   
Being predictive, our results are fundamentally different from
\cite{Gust:Clarke:2004} who gave an ``ANOVA'' like decomposition 
of the posterior variance 
of a parameter in terms of model components.

\subsubsection{The case $K=2$}
\label{K=2}

As an illustration of how our variance decomposition resembles Cochran's Theorem,
we explicitly convert the terms in a three term decomposition to a convex
combination of quadratic forms.  Consistent with the notation of \cite{Draper:1995},
we write $s_i$ to represent `scenarios' $i=1, \ldots , I$ and $m_{ij}$ to represent
models within scenarios, $j=1, \ldots , J$.  In our notation, the $s_i$'s correspond to
the values of $V_1$ andthe  $m_{ij}$'s correspond to values of $V_2$ nested within $V_1$.
Now, Prop.  \ref{General_pred_variance_prop} gives
\begin{align} 
\nonumber 
Var(Y_{n+1} \vert  {\cal{D}}_n) &=  E_{V_1}E_{V_2}Var(Y_{n+1} \vert  {\cal{D}}_n, V_1, V_2) + E_{V_1}Var_{V_2}E(Y_{n+1} \vert  {\cal{D}}_n, V_1, V_2) \\
   & + Var_{V_1}E(Y_{n+1} \vert  {\cal{D}}_n, V_1) \label{two_level_var} \\
\nonumber &=\sum_{i=1}^I p(s_i \vert  {\cal{D}}_n ) \sum_{j=1}^J p( m_{ij} \vert  {\cal{D}}_n, s_i )Var(Y_{n+1} \vert  {\cal{D}}_n , s_i , m_{ij})  \\
\nonumber &+\sum_{i=1}^I p(s_i \vert  {\cal{D}}_n ) \sum_{j=1}^J p( m_{ij} \vert  {\cal{D}}_n, s_i )\left[ E(Y_{n+1} \vert  {\cal{D}}_n , m_{ij}, s_i) - E( Y_{n+1} \vert  {\cal{D}}_n, s_i) \right] ^2 \\
&+\sum_{i=1}^I p(s_i \vert  {\cal{D}}_n ) \left[ E(Y_{n+1} \vert  {\cal{D}}_n , s_i) - E( Y_{n+1} \vert  {\cal{D}}_n) \right] ^2. 
\label{two_level_var_expand}
 \end{align}

For ease of notation, let
\begin{itemize}
\item $p(s_i \vert  {\cal{D}}_n) = \xi_i$
\item $p(m_{ij} \vert  {\cal{D}}_n, s_i) = \omega_{ij}$
\item $E( Y_{n+1} \vert  {\cal{D}}_n) = \bar{y}$
\item $E(Y_{n+1} \vert  {\cal{D}}_n , s_i) = \bar{y}_{i\cdot}$
\item $E(Y_{n+1} \vert  {\cal{D}}_n , m_{ij}, s_i) = \hat{y}_{ij}$.
\end{itemize}
Now we can restate \eqref{two_level_var_expand} as
\begin{align}
\label{T1} Var(Y_{n+1} \vert  {\cal{D}}_n) &= \sum_{i=1}^I \xi_i \sum_{j=1}^J \omega_{ij} Var(Y_{n+1} \vert  {\cal{D}}_n, m_{ij}, s_i)\\
\label{T2} &+\sum_{i=1}^I \xi_i \sum_{j=1}^J \omega_{ij} \left(\hat{y}_{ij} - \bar{y}_{i\cdot} \right) ^2 \\
 \label{T3}&+ \sum_{i=1}^I \xi_i\left( \bar{y}_{i\cdot} - \bar{y} \right) ^2.
\end{align}

Our strategy is to express  each term in $Var(Y_{n+1} \vert  {\cal{D}}_n)$ in vector 
notation so we can recognize quadratic forms. 
First, we see that \eqref{T1} is an expected quadratic form, i.e.
\begin{eqnarray}
\sum_{i=1}^I \xi_i \sum_{j=1}^J \omega_{ij} Var(Y_{n+1} \vert  {\cal{D}}_n, m_{ij}, s_i) = \sum_{i=1}^I \xi_i \sum_{j=1}^J \omega_{ij} E((Y_{n+1} -\hat{y}_{ij})^2 \vert{\cal{D}}_n, m_{ij}, s_i ).
\label{T1Q}
\end{eqnarray}

For \eqref{T2}, write $W_i$ for the column vector 
$W_i = (\sqrt{\omega_{i1}}, \ldots \sqrt{\omega_{iJ}})'$, and write 
$\hat{Y}_i$ for the column vector $\hat{Y}_i = (\hat{y}_{i1} - \bar{y}_{i\cdot}, \ldots, \hat{y}_{iJ} - \bar{y}_{i\cdot})'$.   Now  \eqref{T2} is
\begin{align}
\nonumber \sum_{i=1}^I \xi_i \sum_{j=1}^J \omega_{ij} \left(\hat{y}_{ij} - \bar{y}_{i\cdot} \right) ^2   &=\sum_{i=1}^I \xi_i W'_i \hat{Y}_i\hat{Y}'_iW_i \\
  &=\sum_{i=1}^I \xi_i \hat{Y}'_iW_i W'_i \hat{Y}_i.
\label{T2Q}
\end{align}

Similarly, for term \eqref{T3}, write $S$ for the column vector 
$S=(\sqrt{\xi_1}, \ldots, \sqrt{\xi_I})'$
and $\bar{Y} = (\bar{y}_{1\cdot}-\bar{y}, \ldots,  \bar{y}_{I\cdot}-\bar{y})'$. 
Then we have that \eqref{T3} is
\begin{align}
\nonumber \sum_{i=1}^I \xi_i \left(\bar{y}_{i\cdot}  - \bar{y}_{\cdot \cdot} \right) ^2   &=S' \bar{Y} \bar{Y}'S \\
  &= \bar{Y}'SS'\bar{Y}.
\label{T3Q}
\end{align}
So, using \eqref{T1Q}, \eqref{T2Q}, and \eqref{T3Q}, 
we can rewrite \eqref{two_level_var} as
\begin{align}
    \label{Quad_T1} 
Var(Y_{n+1} \vert  {\cal{D}}_n) &=  \sum_{i=1}^I \xi_i \sum_{j=1}^J \omega_{ij} E((Y_{n+1} - \hat{y}_{ij})^2 \vert {\cal{D}}_n, m_{ij}, s_i )\\
  \label{Quad_T2}&+\sum_{i=1}^I \xi_i \hat{Y}'_iW_i W'_i \hat{Y}_i  \\
 \label{Quad_T3}& + \bar{Y}'SS'\bar{Y}.
\end{align}
Now we see each term in the posterior predictive variance is a quadratic form, i.e., a homogeneous polynomial of
order two, even if the terms in \eqref{Quad_T1}
are (trivial) quadratic forms of dimension one.   

To see how the distributional aspects of \eqref{Quad_T1}, \eqref{Quad_T2}, and 
\eqref{Quad_T3} parallel the distributional statements in Cochran's Theorem,
we proceed as follows.   Note that regarding ${\cal{D}}_n$ as a random variable 
rather than as observed data
means that all terms in the decomposition can also be regarded as random variables.
Next, assume all data are normal.   Now,
\begin{eqnarray}
\nonumber Var(Y_{n+1} \vert  {\cal{D}}_n) 
-  \sum_{i=1}^I \xi_i \sum_{j=1}^J \omega_{ij} E((Y_{n+1} 
 - \hat{y}_{ij})^2 \vert {\cal{D}}_n, m_{ij}, s_i )
 = \sum_{i=1}^I \xi_i \hat{Y}'_iW_i W'_i \hat{Y}_i  
 + \bar{Y}'SS'\bar{Y} \\
\label{Cochrandist1}
\end{eqnarray}
in which each term has a distribution.
We first deal with the two terms on the right in \eqref{Cochrandist1}.

To begin, we recall Theorem 2.1 in \cite{Box:1954} that
generalizes Cochran's theorem for the distribution for 
quadratic forms.  Namely, if $X\sim N(0, V)$, with $V$ a $p \times p$ covariance matrix.  
Then if $Q= X^TMX$ is any real quadratic form of rank $r \leq p$, $Q$ is distributed like a quantity
\begin{equation}
\label{Box_dist}
\sum_{j=1}^r \lambda_j \chi^2_1
\end{equation}
with $r \leq p$  and $\lambda_i$ the $i^{th}$ eigenvalue of $VM$.

Now, look at the first term on the right, and let  $A_{i} =  W_i W_i'$. We know $A_{i}$ is a $J \times J$, symmetric, and semi-positive definite because (\ref{Quad_T2}) is a variance between values $V_1$ within $V_2$ and by definition variances are positive.   

Next, consider the second term on the right and let $B= SS'$  
which is $I \times I$, symmetric and semi-positive definite by definition of variance. 
Further suppose $\bar{Y} \sim N(0, \Sigma^*)$  and  $\sqrt{\xi_i}\hat{Y}_i \sim N(0, \Sigma_i)$.  

Now, since both terms on the right in \eqref{Cochrandist1} are quadratic forms in a normal random vector, we can apply
Theorem 2.1 in \cite{Box:1954} to each of them.   So,
\eqref{Cochrandist1} gives 
\begin{equation}
Var(Y_{n+1} \vert  {\cal{D}}_n) - \sum_{i=1}^I \xi_i \sum_{j=1}^J \omega_{ij} E((Y_{n+1} -\hat{y}_{ij})^2 \vert{\cal{D}}_n, m_{ij}, s_i ) \sim  \sum^I_{i=1} \xi_i \sum^J_{j=1} \lambda_{ij} \chi^2_1 +  \sum^I_{i=1} \lambda_i\chi^2_1
\label{Cochrandist2}
\end{equation}
where $\lambda_i$ is the $i^{th}$ eigenvalue of $B\Sigma^*$ and $\lambda_{ij}$ is the
$j$-th eigenvalues of $A_i \Sigma_i$.  
That is, the two terms on the rightof \eqref{Cochrandisst1}
 are convex and weighted sums, respectively,
of $\chi^2_1$ random variables.

The second term on the left is the expectation of a $\chi^2_1$ random variable.
To see this,  suppose  
$(Y_{n+1} - \hat{y}_{ij} \vert {\cal{D}}_n, m_{ij}, s_i) \sim N(\mu, \sigma^2)$ so that
$((Y_{n+1} - \hat{y}_{ij})^2 \vert {\cal{D}}_n, m_{ij}, s_i) \sim \sigma^2 \chi^2_1(\mu^2)$.  This gives
\begin{eqnarray}
E((Y_{n+1} - \hat{y}_{ij})^2 \vert {\cal{D}}_n, m_{ij}, s_i ) = \mu^2 + \sigma^2
\label{extraterm}
\nonumber
\end{eqnarray}
where  $\mu$ and $\sigma$ depend on $({\cal{D}}_n, m_{ij}, s_i )$.   It
is difficult to determine the distribution of \eqref{extraterm} explicitly but because we are
taking a convex combination of terms like it, computations
suggest it is approximately normal.

Since all three terms in \eqref{two_level_var_expand} are variances and hence
corrected for their means, we regard \eqref{Quad_T1} is a new term that arises from
trying to derive a representation of $Var(Y_{n+1} \vert  {\cal{D}}_n)$ as
an expansion in the form of Cochran's Theorem.  
to complete our analogy, recall
Cochran's Theorem gives as many terms as there are factors plus a residual term.
We get $\dim(V)$ terms, i.e.,  the number of factors, plus an
extra term,
\eqref{Quad_T1}, the predictive analog of the residual term.

If desired, we can approximate distributions of the right hand terms in \eqref{Cochrandist2}
more compactly
by using other results from \cite{Box:1954}. 
Theorem 2.2 gives the formulas for the $i^{th}$ cumulant of \eqref{Box_dist} as
$$
Q_i = 2^{i-1}(i-1)! \sum^{r}_{j=1} \lambda_j.
$$

Using this, we can approximate \eqref{Box_dist} by $g\chi^2(h)$ where
$$
g = \frac{1}{2}\frac{Q_1^2}{Q_2} = \frac{ \sum \lambda^2_j}{\sum \nu_j \lambda_j}
$$
and
$$
h = \frac{2 Q_1^2}{Q_2} = \frac{(\sum \lambda_j)^2}{\sum  \lambda_j^2}.
$$
Box gives this approximation in part because it has the same first two moments
as \eqref{Box_dist}.  Box also notes that when all $\lambda_j$ are equal, 
the degrees of freedom, $h$, is smaller than appropriate.  

Using this we can approximate
$
\bar{Y}' B \bar{Y} =\bar{Y}' SS'\bar{Y} 
$
by
 \begin{equation}
\label{T3_approx}
  g\chi^2_h = \frac{ \sum \lambda^2_i}{\sum \lambda_i} \chi^2\left(\frac{(\sum \lambda_i)^2}{\sum \lambda_i^2} \right).
\end{equation}
Also, we can approximate
$$
\sqrt{\xi_i }\hat{Y}'_i A_i \sqrt{\xi_i }\hat{Y}_i = \sqrt{\xi_i }\hat{Y}'_i  W_i W'_i  \sqrt{\xi_i }\hat{Y}_i
$$
by
$$
 g_i \chi^2_{h_i} =\frac{ \sum_j \lambda^2_{ij}}{\sum_j \lambda_{ij}} \chi^2\left(\frac{(\sum_j \lambda_{ij})^2}{\sum_j \lambda_{ij}^2} \right).
 $$
Hence,  we have the approximate distribution
\begin{equation}
\nonumber Var(Y_{n+1} \vert  {\cal{D}}_n) -  \sum_{i=1}^I \xi_i \sum_{j=1}^J \omega_{ij} E((Y_{n+1} -\hat{y}_{ij})^2 \vert{\cal{D}}_n, m_{ij}, s_i )\overset{approx}{\sim} g_i \chi^2_{h_i} +  g\chi^2_h.
\end{equation}

In classical Cochran's Theorem settings, the $\chi^2$ distributional results are used 
to form $F$-tests.  Here,this is not readily feasible because the quadratic forms are
not in general independent, the matrices in them are not idempotent, and
 we do not have a definite distribution for the second term on the left in \eqref{Cochrandist1}.
Our point here has been only to show the parallel between the Cochran's Theorem
decomposition and our posterior predictive variance decomposition.
In practice,  instead of $F$ test, our decompsotion leads to bootstrap 
tests that we present in Subsec. \ref{testing}.

\subsubsection{General $K$}
\label{generalK}

Deriving quadratic forms and distributional expressions
for $Var(Y_{n+1} \vert {\cal{D}}_n)$ for general $K$ is similar to the derivation of
\eqref{Cochrandist1} and \eqref{Cochrandist2}, respectively, seen in Subsec.\ref{K=2}.
For the sake of completeness, we state these two results below.

Our first result gives the general expression for the posterior predictive variance
in terms of quadratic forms.  Let 
$$ 
\hat{y}_{v_{i_1},\ldots,v_{i_{k}}} = E\left(Y_{n+1} \vert v_{i_1},\ldots,v_{i_{k}} \right).
$$
We have the following.
\begin{proposition}
\label{pred_cochran}
For a $K$-factor predictive scheme, the posterior predictive
variance can be written as a sum of weighted quadratic forms as follows:
\begin{align}
\nonumber Var(Y_{n+1} \vert  {\cal{D}}_n)
&=  \sum_{i_{1}=1}^{I_{1}} p\left(v_{i_1}\vert {\cal D}_n\right) \ldots  \sum_{i_{K}=1}^{I_{K}} p\left(v_{i_K}\vert {\cal D}_n,v_{i_1},\ldots,v_{i_{K-1}}\right)  E\left(\left(Y_{n+1} - \hat{y}_{v_{i_1},\ldots,v_{i_{K}}} \right)^2 \vert{\cal{D}}_n,v_{i_1} \ldots , v_{i_K} \right) \\ \nonumber
&+   \sum_{i_{1}=1}^{I_{1}} p\left(v_{i_1}\vert {\cal D}_n\right) \ldots  \sum_{i_{K}=1}^{I_{K-1}} p\left(v_{i_{K-1}}\vert {\cal D}_n,v_{i_1},\ldots,v_{i_{K-2}}\right) \hat{Y}'_{K,\ldots,1} A_{K,\ldots,1} \hat{Y}_{K,\ldots,1}\\ \nonumber
&+  \sum_{i_{1}=1}^{I_{1}} p\left(v_{i_1}\vert {\cal D}_n\right)\ldots  \sum_{i_{K-2}=1}^{I_{K-2}} p\left(v_{i_{K-2}}\vert {\cal D}_n,v_{i_1},\ldots,v_{i_{K-3}}\right)\hat{Y}'_{K-1,\ldots,1} A_{K-1,\ldots,1} \hat{Y}_{K-1,\ldots,1}\\ \nonumber
& \hspace{.6 in} \vdots \hspace{.6 in} \vdots  \hspace{.6 in} \vdots \\ \nonumber
&+   \sum_{i_{1}=1}^{I_{1}} p(v_{i_1}\vert {\cal D}_n)  \hat{Y}'_{2,1} A_{2,1} \hat{Y}_{2,1} \\
&+  \hat{Y}'_{1} A_{1}  \hat{Y}_{1},
\label{sum_quad_from_general}
\end{align}
where
\begin{equation}
\label{sum_general_chi_squared}
A_{k,\ldots,1} =  W_{k,\ldots,1} \left(W_{k,\ldots,1} \right)',
\end{equation}
$$
W_{k,\ldots,1}  = \left(\sqrt{ p\left(v_{i_{k}=1}\vert {\cal D}_n,v_{i_1},\ldots,v_{i_{k-1}}\right)} , \ldots ,\sqrt{ p\left(v_{i_{k}=I_k}\vert {\cal D}_n,v_{i_1},\ldots,v_{i_{k-1}}\right)}  \right),
$$
 and $\hat{Y}_{k,\ldots,1} $ is the column vector of mean adjusted predictions for factor $V_k$ conditional on factors $V_{1}, \ldots V_{k-1}$.  That is, we write
$$
\hat{Y}_{k,\ldots,1} = \left(\left(\hat{y}_{v_{i_1}, \ldots, v_{i_k=1}} - E(Y_{n+1} \vert {\cal D}_n, v_{i_1}, \ldots, v_{i_{k-1}})  \right), \ldots, \left(\hat{y}_{v_{i_1}, \ldots, v_{i_k=I_k}}  - E(Y_{n+1} \vert {\cal D}_n, v_{i_1}, \ldots, v_{i_{k-1}} ) \right)\right)'
$$
where $\hat{y}_{v_{i_1}, \ldots, v_{i_k=j}} = E(Y_{n+1} \vert {\cal D}_n, v_{i_1}, \ldots, v_{i_k=j})$.
\end{proposition}

Note that for the stacking version we replace the posterior probabilities with stacking weights.

Our second result gives the distributions for $K$ of the terms in our expansion
for the posterior predictive variance.  As before, we get sums of $\chi^2_1$ random variables.
We have the following.
\begin{proposition}
\label{weight_chi_sq_prop}
Let  $\hat{Y}_{1} \sim N(0, \Sigma)$  and  $W_{k,\ldots,1} \hat{Y}_{k,\ldots,1} \sim N(0, \Sigma_{k,\ldots,1})$. Then the sum of quadratic forms in (\ref{sum_quad_from_general}) are distributed like a sum of weighted 
$\chi$-squared random variable as follows
\begin{align}
\label{sum_quad_chi_sqaures_general}
\nonumber Var(Y_{n+1} \vert  {\cal{D}}_n)&\sim  \sum_{i_{1}=1}^{I_{1}} p\left(v_{i_1}\vert {\cal D}_n\right) \ldots  \sum_{i_{K}=1}^{I_{K}} p\left(v_{i_K}\vert {\cal D}_n,v_{i_1},\ldots,v_{i_{K-1}}\right)  E((Y_{n+1} -\hat{y}_{v_{i_1},\ldots,v_{i_{K}}})^2 \vert{\cal{D}}_n,v_{i_1},\ldots,v_{i_{K}} )\\
\nonumber & + \sum_{i_{1}=1}^{I_{1}} p\left(v_{i_1}\vert {\cal D}_n\right) \ldots \sum_{i_{K-1}=1}^{I_{K-1}}  p\left(v_{i_{K-1}}\vert {\cal D}_n,v_{i_1},\ldots,v_{i_{K-2}}\right) \sum_{i_{K}=1}^{I_{K}}  \lambda_{K \ldots, 1} \chi^2_1  \\
\nonumber &+ \sum_{i_{1}=1}^{I_{1}} p\left(v_{i_1}\vert {\cal D}_n\right) \ldots \sum_{i_{K-2}=1}^{I_{K-2}}  p\left(v_{i_{K-2}}\vert {\cal D}_n,v_{i_1},\ldots,v_{i_{K-3}}\right)  \sum_{i_{K-1}=1}^{I_{K-1}}  \lambda_{K-1 \ldots, 1} \chi^2_1  \\
\nonumber & \hspace{.6 in} \vdots \hspace{.6 in} \vdots \hspace{.6 in} \vdots  \\
\nonumber &+\sum_{i_{1}=1}^{I_{1}} p\left(v_{i_1}\vert {\cal D}_n\right) \sum_{i_{2} =1}^{I_{2}} \lambda_{2 , 1 } \chi^2_1  \\
 &+\sum_{i_{1}=1}^{I_{1}} \lambda_{1} \chi^2_1
 \end{align}
where $\lambda_{k, \ldots ,1}$ is the 
$k$th eigenvalue of $A_{k,\ldots,1} \Sigma_{k,\ldots,1}$.
\end{proposition}

\subsection{Testing}
\label{testing}

The $\chi$-squared distributions derived
at the end of Subsec. \ref{K=2} or motivated by Prop. \ref{weight_chi_sq_prop}
are analogous (apart from dependence and normality) to the 
distributional result from Cochrane's theorem.  In the ANOVA context, it is common
to test the equality of levels of a factor.  Here, 
the corresponding null hypothesis would be the equality of
expectations of the predictive distributions within a factor or the posterior weight
being close to one for a single level within a factor.  Here,  we rephrase these tests 
as a way to determine the relative importance
of terms in our decomposition.   

Specifically, we want to test whether a term in the variance decomposition is
a substantial fraction of the overall variance.  Consider the case $K=1$ that gives
a two-term decomposition for $Var(Y_{n+1} \vert {\cal{D}}_n)$.
Now, we want to test hypotheses of the form
$$
H_0: E\left( \frac{Var_{V_1}(Y_{n+1}\vert {\cal{D}}_n, V_1)}{Var(Y_{n+1} \vert  {\cal{D}}_n) }\right) \geq \tau
$$
$$
H_1: E\left( \frac{Var_{V_1}(Y_{n+1}\vert {\cal{D}}_n, V_1)}{Var(Y_{n+1} \vert  {\cal{D}}_n )}\right) < \tau.
$$
for some pre-selected value of $\tau >0$.
Since we do not have a likelihood for the argument of the expectation in $H_0$,
we are led to a nonparametric test based on bootstrapping.  

Assuming that the data is representative of of the DG,
we use bootstrapping on the argument of the expectation in $H_0$.
The result is a data set of the form
$$
Z_b =  \frac{Var_{V_1}E(Y_{n+1}\vert {\cal{D}}^b_n, V_1)}{Var(Y_{n+1} \vert  {\cal{D}}^b_n)},
$$
for $b=1, \ldots, B$ that can be regarded as
representative of
$\frac{Var_{V_1}(Y_{n+1}\vert {\cal{D}}_n, V_1)}{Var(Y_{n+1} \vert  {\cal{D}}_n) }$
as a random variable.   Writing $\bar{z}$ and $SE(\bar{z})$ for the mean and its
standard error for the $Z_b$'s we form
$$
t = \frac{\bar{z} -\tau}{SE(\bar{z})}.
$$
We use $\tau$ in this expression because it corresponds to seeking the uniformly most
powerful test for $H_0$.  Note that $\bar{z}$ is (mild) abuse of notation.
In fact, we should write the $Z_b$'s with `hats' over the variances and expectations since
we are bootstrapping.  We see this as a point to bear in mind but do not wish to clutter
the notation.

Let $J > B$.    In a second layer of bootstrapping, draw $J$ samples of size
$B$ from $z_1,\ldots, z_B$, with replacement.     Denote these by
$z^\prime_1, \dots, z^\prime_J$ where each $z^\prime_j$ has $B$ entries.
To get a distribution for $T=t$ as a random variable
under the null, we generate the vectors
$$
\tilde{z^\prime}_j = z^\prime_j - (\bar{z^\prime}_j -\tau) {\bf 1}_B
= z^\prime_j - \left(\frac{1}{B}\sum^B_{b=1}z_{j,b}^\prime -\tau \right) {\bf 1}_B
$$
where ${\bf 1}_B = (1, \ldots, 1)$ is $B$-dimensional.  Now, we have $J$ different samples for
which the mean is $\tau$.  From the samples corrected by their means and $\tau$ so they satisfy the null, we form the $t$-statistics
$$
\tilde{t}_j = \frac{\bar{\tilde{z}}^\prime_j-\tau}{SE(\bar{\tilde{z}}^\prime_j)}
$$
for $j=1, \ldots, J$
and calculate the estimated achieved significance level,
$$
\widehat{ASL} = \frac{1}{J}\sum I(\tilde{t}_j  \leq t ).
$$
When the $\widehat{ASL}$ is small, we reject $H_0$ and this tells us that
$Var_{V_1}E(Y_{n+1}\vert {\cal{D}}_n, V_1)\approx 0$ suggesting that  
$E(Y_{n+1}\vert {\cal{D}}_n, V_1)$ is constant in $V_1$. Therefore, omitting
this term in forming the PI for $Y_{n+1}$ does not affect the width.
Here, when we do this testing, we default to a threshold of .05 for the ASL
for convenience.

This bootstrapping approach allows us to move beyond the assumption
that the predictions follow a normal distribution as used in the discussion at the end of 
Subsec. \ref{K=2} and in Prop. \ref{weight_chi_sq_prop}.

\section{Revisting Draper (1995)}
\label{challenger}

Here we apply our techniques to two examples given in \cite{Draper:1995} and one further
example that his second example motivates.
The first example involves predicting the price of oil; the second example
involves predicting the chance of failure of O-rings in a space shuttle at a new temperature.
Our third example for this section is an extension of the latter data type with
a more difficult variable selection problem.  
Draper's main point was when making predictions, we need to consider the uncertainty of the `structural' choices we make or we can be lead to bad decisions.    Here, we have
formalized Draper's concept of structural choices in our conditioning variable $V$.
One danger in poor structural choices is that a PI may be found that is
unrealistically small leading to over-confidence.

By using the testing procedure in Subsec. \ref{testing},  we are able to determine
which terms in the Cochran-like decomposition (see Clause (i)
in Prop. \ref{General_pred_variance_prop})
can be ignored.
That is, our test is able to determine if a structural choice should or should
not be included in the uncertainty analysis of the predictive distribution.

\subsection{Oil Prices}

In the oil prices example in \cite{Draper:1995} there are two structural components to
the modeling namely,  12 economic scenarios with 10 economic models nested inside them.
These components represent 120 models and hence
introduce model uncertainty that must be quantified to generate
good PI's.  

In Draper's analysis each model was used given the parameters of each scenario.
This corresponds to $K =2$ and a three term posterior predictive variance decomposition.
Let $s_i$ denote scenario $i$ and $m_{ij}$ be model $j$ within scenario $i$.
Write $s_i \in S$ and $m_{ij} \in M_i \subset M$ where $M_i$ is the set of models
for scenario $i$ and $M$ is the union of the $M_i$'s.   Now, we have
\begin{align}  
\nonumber Var(Y_{n+1} \vert  {\cal{D}}_n) (S,M)  &=E_S E_MVar(Y_{n+1} \vert  {\cal{D}}_n , S , M)  \\
\nonumber &+ E_S Var_M(E(Y_{n+1} \vert  {\cal{D}}_n , S, M) )\\
&+ Var_S(E(Y_{n+1} \vert  {\cal{D}}_n , S)).
\label{DraperOil}
 \end{align}
The corresponding decomposition given by Draper is
$\nonumber Var(Y_{n+1} \vert  {\cal{D}}_n) = 178 + 363 + 354 = 895$.  We
cannot recompute this example because neither the data nor the details on the scenarios or
models are available to us.   However, in this case it is seen that
the between-scenarios variance, i.e., term \eqref{DraperOil}, contributes about 40\% to
the posterior predictive variance.   The second term on the right,  the between-models
within scenarios variance., is also about 40\%  The variance attributable to the predictions
within models and scenarios is about 20\%.    (See \Cref{vartable} for the definition of terms.)
Thus all the three terms must
be used when forming PI's.  We surmise therefore that if we had the original data
and could therefore perform the desired hypothesis tests, we would not reject
any of the null hypotheses.

\subsection{Challenger Disaster}
\label{disaster}

Making the decision
to launch the space shuttle at an ambient temperature at which the
various components had not been tested ended up being catastrophic
-- and could have been avoided had a proper uncertainty analysis had been done.  
Statistically, the error of the decision makers was to choose a single model from a
model list rather than incorporating all sources of predictive uncertainty into their analysis.
The goal of this example originally was to show that a correct analysis
of the various sources of uncertainty would have led to a PI for $p_{t=31}$,
the probability of
an O-ring failure (at $31^\circ$) of $(.33,1]$ i.e., too high for a launch to
be safe.
Our goal in re-analyzing Draper's example is to identify which sources of
uncertainty can be neglected.    

We have 23 observations of the number of damaged O-rings ranging from zero to
six (because each shuttle had six O-rings).    Each observation also has a temperature $t$
and a `leak-check' pressure $s$.   Following Draper's analysis we also use $t^2$
as an explanatory variable.  Thus we have 24 vectors, each of length four.

We assume the number of damaged O-rings follows a $Binomial(6, p)$ distribution
where $p$ is a function of the explanatory variables via one of three link functions,
logit, $c\log\log$, and probit.  Thus, we have structural uncertainty in the choice
of variables and in the choice of link function.  In our notation, we set
$V_1 = \{L, C, P \}$ for the choice of
link function, logit, $c\log\log$, and probit respectively. Also let
$V_2 = \{t, t^2, s, \text{no effect}\}$ where no effect means an intercept-only model.  
The 24 models are summarized in \Cref{Tab_challenger2}.

\begin{table}[ht]
\caption{\textbf{List of models for the Cahllenger disaster data:}  This table lists
all 24 models under consideration broken down by their structural choices -- link functions
and explanatory variables.}
\label{Tab_challenger2}
\centering
 \resizebox{\textwidth}{!}{  
\begin{tabular}{ccccccccccccccccccc}
  \hline
 ${\cal{V}}^{(2)}$ &  $m_1$& $m_2$ & $m_3$ & $m_4$ & $m_5$ & $m_6$ & $m_7$& $m_8$ & $m_{9}$ & $m_{10}$ & $m_{11}$  & $m_{12}$ & $m_{13}$\\
  \hline
$V_1$   &  L & L & L & L & L  & L & L&L   &  C & C & C & C & C    \\
$V_2$  & $t$ &  $t^2$& $s$ & $t, t^2$ & $t, s$ & $t^2, s$ & $t, t^2, s$   & no effect  & $t$ &  $t^2$& $s$ & $t, t^2$ & $t, s$ \\
   \hline

  \hline
 ${\cal{V}}^{(2)}$  & $m_{14}$ & $m_{15}$ & $m_{16}$& $m_{17}$& $m_{18}$ & $m_{19}$ & $m_{20}$ & $m_{21}$ & $m_{22}$& $m_{23}$& $m_{24}$  \\
  \hline
$V_1$  & C&C & C  &  P & P & P & P & P&P  & P  &P   \\
$V_2$  & $t^2, s$ & $t, t^2, s$   & no effect & $t$ &  $t^2$& $s$ & $t, t^2$ & $t, s$& $t^2, s$  & $t, t^2, s$   & no effect  \\
   \hline
\end{tabular}
}
\end{table}

In fact, Draper did not consider all of these models.   Essentially he put zero
prior probability on all models except for $m_1, m_4, m_5,m_7,m_8$, and $m_{15}$.
Accordingly,  he only considered the set
$$
 {\cal{M}} = \{m_1, m_4, m_5, m_7, m_8, m_{15} \}
$$
with a uniform prior.
Draper then gave a table of posterior quantities for the structural choices, and a
posterior predictive variance decomposition for within-structure and between-structure
variances as
\begin{eqnarray}
Var(p_{t=31}\vert  {\cal{D}}_{23})= Var_{within} + Var_{between} = 0.0338 + 0.0135 = 0.0473.
\label{twotermDraper}
\end{eqnarray}
That is, even though there were two structural choices, Draper used a decomposition
appropriate for one.  This corresponds to using our result Clause (ii) in
Prop. \ref{General_pred_variance_prop}.  Draper's conclusion was that
$.0135/.0473 \approx 28.5\%$ so the uncertainty represented by the second term
in \eqref{twotermDraper} could not be neglected.

Here we extend Draper's analysis and confirm that structural uncertainty should not be ignored.
For our implementation, we use the full set of 24 models and do not employ the same
approximations.   Then, we use the {\sf{BMA}} package in R to get
the posterior distributions of the parameters of the models and the posterior weights
for $V_2$.  We also use the ${\sf rjmcmc}$ package to get the posterior weights
for $V_1$.
We note in passing that the resulting posterior distributions were qualitatively similar to
Draper's approximate posteriors.

Considering all sources of uncertainty yields a posterior predictive variance decomposition of
\begin{align}
\label{var_challenger_3term}
\nonumber Var(p_{t=31}\vert  {\cal{D}}_{23})&=  E_{V_1}E_{V_2}Var(p_{t=31} \vert  {\cal{D}}_{23}, V_1, V_2) + E_{V_1}Var_{V_2}E(p_{t=31} \vert  {\cal{D}}_{23}, V_1, V_2) \\ \nonumber
   & + Var_{V_1}E(p_{t=31} \vert  {\cal{D}}_{23}, V_1) \\ \nonumber
   &=  0.01469 + 0.0996 + 0.0017 \\
   &= 0.11599.
\end{align}
This is almost three times the variance as obtained by Draper.  We confirm his intuition that
structural uncertainty was much greater than assumed when making the decision
to launch the shuttle.   Moreover, Draper commented that other analyses could lead to larger
posterior variances.  So, \eqref{var_challenger_3term} is consistent with his intuition.

We can go beyond Draper's analysis by testing the terms in \eqref{var_challenger_3term}.
With $\tau = .05$, the hypotheses for testing whether the between
link functions variance is a substantial
portion of the posterior variance are
$$
H_0:  E\left( \frac{Var_{V_1}E(p_{t=31} \vert  {\cal{D}}_{23}, V_1)}{Var(p_{t=31} \vert  {\cal{D}}_{23})}\right) \geq 0.05
$$
versus
$$
H_1:  E\left( \frac{Var_{V_1}E(p_{t=31} \vert  {\cal{D}}_{23}, V_1)}{Var(p_{t=31} \vert  {\cal{D}}_{23})}\right) < 0.05.
$$
The test statistic for this test is
$$
\bar{z}_1 =  \frac{Var_{V_1}E(p_{t=31} \vert  {\cal{D}}_n, V_1)}{Var(p_{t=31}\vert  {\cal{D}}_{23})} = \frac{0.0017 }{0.11599} = 0.0147.
$$
The test described in Subsec. \ref{testing} gives an estimated achieved significance level
$\widehat{ASL}_1 = 0$. Thus, we conclude there is essentially no
between-link functions variance and can
ignore this term in the posterior predictive variance.

Next we test the between-models within-link functions term. Here the hypotheses are  
$$
H_0:  E\left( \frac{ E_{V_1}Var_{V_2}E(p_{t=31}\vert  {\cal{D}}_{23}, V_1, V_2) }{Var(p_{t=31}\vert  {\cal{D}}_{23})}\right) \geq 0.05
$$
versus
$$
H_1:  E\left( \frac{E_{V_1}Var_{V_2}E(p_{t=31}\vert  {\cal{D}}_{23}, V_1, V_2) }{Var(p_{t=31} \vert  {\cal{D}}_{23})}\right) < 0.05.
$$
The test statistic for this test is
$$
\bar{z}_2 =  \frac{ E_{V_1}Var_{V_2}E(p_{t=31} \vert  {\cal{D}}_n, V_1, V_2) }{Var(p_{t=31} \vert  {\cal{D}}_n)} = \frac{0.0996 }{0.11599} = 0.86 .
$$
Since the estimated contribution of the posterior predictive variance
from the between models within link functions variance is 86 percent, we can safely
assume this is a significant source of uncertainty.
More formally, we find that this test has an estimated achieved significance level
$\widehat{ASL}_2 = 1$, confirming our intuition.

Finally, we test the between-predictions within-models and link functions term.
The hypotheses are
$$
H_0:  E\left( \frac{ E_{V_1}E_{V_2}Var(p_{t=31} \vert  {\cal{D}}_{23}, V_1, V_2) }{Var(p_{t=31} \vert  {\cal{D}}_{23})}\right) \geq 0.05
$$
versus
$$
H_1:  E\left( \frac{ E_{V_1}E_{V_2}Var(p_{t=31}\vert  {\cal{D}}_{23}, V_1, V_2) }{Var(p_{t=31} \vert  {\cal{D}}_{23})}\right) < 0.05.
$$
The test statistic is
$$
\bar{z}_3 =  \frac{E_{V_1}E_{V_2}Var(p_{t=31} \vert  {\cal{D}}_n, V_1, V_2)}{Var(p_{t=31} \vert  {\cal{D}}_n)} = \frac{ 0.0147}{0.11599} = 0.127
$$
which gives $\widehat{ASL}_3 = 1$, i.e., non-rejection of the null.

Overall,  we conclude that the terms representing the between-models within-link functions variance
and the between-predictions within-models and links variance are terms that must be retained.

\subsection{Simulated binomial example}
\label{simbinom}

In this section, we study a simulated example following the same structure as the Challenger data.
That is we simulated $n$ observations from a binomial generalized linear model
$$
Y_i \vert p_i \sim Binomial(30,p_i),
$$
where $p_i = \frac{1}{1+e^{-X'_i\beta}}$, in which $X'_i \sim N(0,1)$ is a $1 \times 10$ 
vector of explanatory variables, and 
$$
\beta = (0.75, 0.25, -0.3,0.5,0,0,0,0,0,0)
$$ 
is a 
$10 \times 1$ vector of true regression parameters. Here we let 6 entries in $\beta$ be zero to represent a meaningful model selection problem. In this problem, we again recognize 
three sources of structural uncertainty: predictive uncertainty within-models and link functions 
(`predictions'),
models within link functions (`models', $V_2$), and link functions (`links', $V_1$).  Our goal 
with this example is to 
study the effect of the sample size on each each term in the posterior 
predictive variance decomposition,  as well as each of the test statistics.

We continue to use a three term decomposition like that in Subsec.  \ref{disaster}:
\begin{eqnarray}
\nonumber Var(Y_{n+1}\vert  {\cal{D}}_{n})=  E_{V_1}E_{V_2}Var(Y_{n+1} \vert  {\cal{D}}_{n}, V_1, V_2) + E_{V_1}Var_{V_1}E(Y_{n+1} \vert  {\cal{D}}_{n}, V_1, V_2) + Var_{V_1}E(Y_{n+1} \vert  {\cal{D}}_{n}, V_1) \\
\label{simbinomdecomp}
\end{eqnarray}
but our `$Y_{n+1}$' here is the number of successes in 30 trials, a random variable, as opposed to
a probability such as $p_{t=31}$.  Thus, the three forms of null hypotheses we want to test are
\begin{eqnarray}
H_0:  & E\left( \frac{Var_{V_1}E(Y_{n+1}\vert  {\cal{D}}_n, V_1)}{Var(Y_{n+1} \vert  {\cal{D}}_{n})}\right) \geq \tau, 
\quad H_0:  E\left( \frac{ E_{V_1}Var_{V_2}E(Y_{n+1} \vert  {\cal{D}}_n, V_1, V_2)  }{Var(Y_{n+1}\vert  {\cal{D}}_{n})}\right) \geq \tau, \nonumber \\ 
&\hbox{and} \quad
H_0:  E\left( \frac{ E_{V_1}E_{V_2}Var(Y_{n+1} \vert  {\cal{D}}_n, V_1, V_2)}{Var(Y_{n+1}\vert  {\cal{D}}_{n})}\right) \geq \tau \nonumber
\label{nulls}
\end{eqnarray}
for $\tau = .01, .05,$ and $.1$.  The corresponding test statistics are
\begin{eqnarray}
\bar{z}_1 =  &\frac{Var_{V_1}E(Y_{n+1}\vert  {\cal{D}}_n, V_1)}{Var(Y_{n+1}\vert  {\cal{D}}_{n})}, \quad
\bar{z}_2 =  \frac{ E_{V_1}Var_{V_2}E(Y_{n+1} \vert  {\cal{D}}_n, V_1, V_2) }{Var(Y_{n+1}\vert  {\cal{D}}_{n})}  \nonumber \\
&\hbox{and} \quad \bar{z}_3 =  \frac{E_{V_1}E_{V_2}Var(Y_{n+1} \vert  {\cal{D}}_n, V_1, V_2)}{Var(Y_{n+1}\vert  {\cal{D}}_{n})}.
\label{teststats}
\nonumber
\end{eqnarray}
To search for patterns in the testing procedure, we compiled the results of our tests for
various choices of $n$, $\tau$, and $T_j$ in \Cref{proplink} along with
the corresponding $\widehat{ASL}$ values.

We see that in column four the value of $\bar{z}_1$ generally
increases with sample size meaning we are less and less likely to reject its null.
In column three,  $\bar{z}_2$ is generally decreasing meaning we are more and more likely
to reject its null.   The second column shows stable inferences over sample size.
Taken together, \Cref{proplink} suggests that as sample size increases,
the link functions proportionately contribute more and more
to the overall variance where as the models contribute less and the 
predictions are stable.  We comment that the overall variance actually
decreases with sample size so the relative importance of, say, links, may
increase even as its absolute importance decreases.

\begin{table}[H]
\caption{\textbf{Results from testing in simulations}: The first column gives the sample size.
The second column corresponds to testing with $\bar{z}_3$ for $\tau = 0.01,0.05,0.1$.
The third and fourth columns are the same but for $\bar{z}_2$ and $\bar{z}_1$.
For instance, the entry $0.53 (1,1,1)$ means we have used $\bar{z}_3=.53$ to test its null
and the three ASL's for the values of $\tau$ were ones, indicating non-rejection.}
\label{proplink}
\centering
\begin{tabular}{rrrr}
  \hline
 $n$ & Predictions & Models & Links \\
  \hline
20 & 0.53 (1,1,1) & 0.43 (1,1,1) & 0.04 (1,0,0) \\
  30 & 0.64 (1,1,1) & 0.31 (1,1,1) & 0.05 (1,0.014,0) \\
  40 & 0.69 (1,1,1) & 0.24 (1,1,1) & 0.07 (1,0.39,0) \\
  50 & 0.77 (1,1,1) & 0.14 (1,0.99,1)& 0.09 (1,1,0.009) \\
  60 & 0.79 (1,1,1) & 0.11 (1,1,0.10) & 0.10 (1,1,0.52) \\
  70 & 0.82 (1,1,1) & 0.09 (1,0.97,10)  & 0.09 (1,1,0.98)\\
  80& 0.82 (1,1,1) & 0.07 (1,.303,0) & 0.15 (1,1,1) \\
 90 & 0.80 (1,1,1) & 0.04 (1,.15,0)  & 0.16(1,1,1) \\
  100 & 0.79 (1,1,1) & 0.04  (1,0.034,0)  & 0.17(1,1,1) \\
  110 & 0.78 (1,1,1) & 0.04  (1,0,0)  & 0.18(1,1,1) \\
  120 & 0.77 (1,1,1) & 0.04  (1,0,0) & 0.19 (1,1,1)\\
   \hline
\end{tabular}
\end{table}

These conclusions are reinforced by \Cref{simulated_binom}.
On the left panel we see that all four variances decrease with $n$, the
top curve representing the sum of the three lower curves.
The right panel shows that as expected the relative contribution of models
decreases monotonically.  It also shows that as $n$ increases, the curves for
links and predictions approach each other.  In simulation results not
shown here, the two curves actually cross around $n=475$ and suggest that
by $n=900$ or so that the curve for predictions will indicate a relatively small
contribution to the decreasing total variance curve compared to the relative 
contribution of links.  However, by this point, the total variance is so small
that the relative contributions of the terms does not matter much.
\begin{figure}[H]
\begin{center}
\begin{tabular}{cc}
\includegraphics[width=.45\columnwidth]{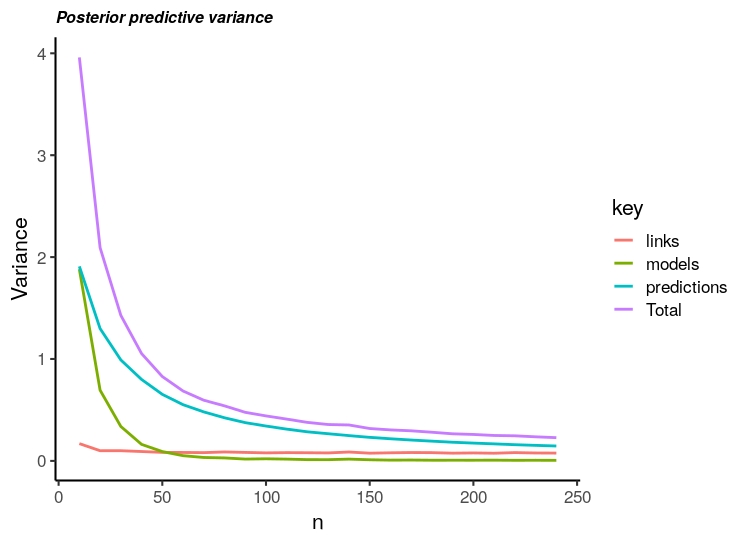}
\includegraphics[width=.45\columnwidth]{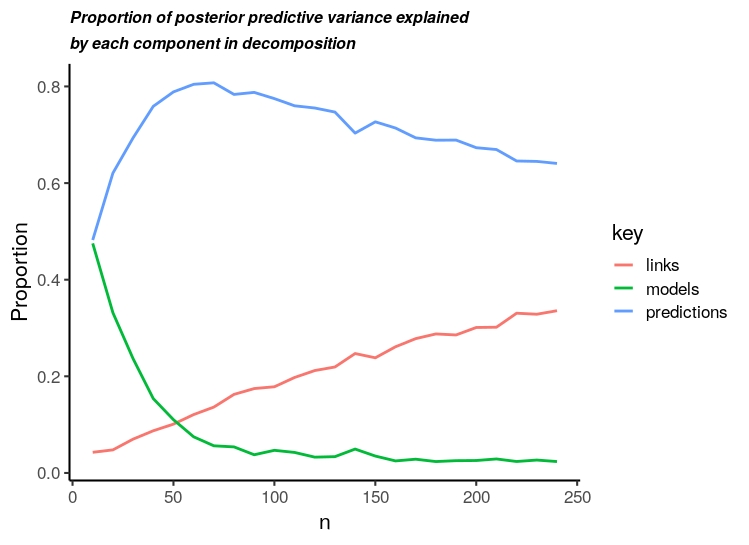}
\end{tabular}
\end{center}
\caption{\textbf{Plots of posterior predictive variances as a function of $n$
 for each term in the decomposition. } Left: The actual values of the terms in the total posterior
predictive variance.  Right: The proportions each term contributes to the total variance.
 Both sets of curves show the results for link function, models, and predictions.}
\label{simulated_binom}
\end{figure}

\section{Example:  Superconductivity Data}
\label{superconduct}

In this section, we analyze the data set {\sf Superconductivity} presented in \cite{Hamidieh:2018}.    
This data set has 81 explanatory variables of a physical or chemical nature to explain a response $Y$
representing temperature measurements (in degrees K) for when a compound
begins to exhibit superconductivity. The full data set has $n= 21263$, and we assume the relationship between $Y$ and the explanatory variables follows a signal plus noise structure, i.e.
$$ 
Y_i = f(X_i) + \epsilon_i 
$$  
for $i=1,\ldots , n$ and where $\epsilon_i \sim N(0,\sigma^2)$.
 \cite{Hamidieh:2018} used a linear model (LM) as a `benchmark model' and then
improved on it by developing an XGBoosting model -- a boosted, penalized tree model. 
The goal in their paper
was to minimize predictive error on a hold out set.  so, they did not consider 
the variance of predictive distributions.

Here we use 5 common predictive models; $m_1 =$ LM , $m_2 =$ neural nets (NN), $m_3 =$ projection pursuit regression (PPR), $m_4 =$ support vector machine with a radial kernel 
(SVM), and $m_5 =$ XGBoosting (XGB). Hence we write $V = (m_{1}, \ldots, m_{5})$. We 
note that these models do not correspond to a probability distribution (except for LM), 
but upon examining the residuals from the other fitted models, we confirmed that the
residuals were normally distributed.  So, we use a normal to form a predictive 
distribution for each of the models.  Moreover, to form the predictive distribution for each model
we fit the model using $n$ data points, and used the $n+1$ observation to predict
$Y_{n+1}$.  

Let the predictor from model $k$ be $\hat{f}_k$, $k=1, ..., 5$.
Then the next outcome is normally distributed, centered at the point predictor $\hat{f}_k(X_{n+1})$ 
with estimated variance
$$
\widehat{Var}\left(Y_{n+1} - \hat{f}_k(X_{n+1})\right) = \widehat{Var}(\hat{f}_k(X_{n+1})) + \widehat{Var}(\hat{\epsilon}_k).
$$
We calculated $\widehat{Var}(\hat{f}_k(X_{n+1}))$ by bootstrapping.
That is, we found a bootstrap distribution
for it and then took its variance.   For $\widehat{Var}(\hat{\epsilon}_k)$, we found the variance
of the residuals from the fitted model.

Formally, the predictive distribution for each model is given by
$$
p(Y_{n+1} \vert m_k) =  N\left(\hat{f}_k(X_{n+1}), \widehat{Var}(\hat{f}_k(X_{n+1})) + \widehat{Var}(\hat{\epsilon}_k) \right).
$$
Since these models are implemented in a frequentist sense, we use stacking to average over the models. The stacked predictive distribution for $Y_{n+1}$ is
$$
Y_{n+1} \sim \sum^5_{k=1} \hat{w}_k({\cal D}_n) p(Y_{n+1} \vert m_k).
$$
Now we present two examples, one where we randomly sample $500$ of the data points to form the predictive distributions and test whether the between-models variance is important, and another where we use the whole data set and perform the same test.
We will see that with the smaller sample size, the between-models variance term in
the decomposition using $V$ contributes about two-thirds of the total posterior predictive variance.
However, when the full data set is used, the estimated contribution from the between-models
term drops to about 4\%.

First we took a random sample of 500 observations from the whole data set. 
We set $B=200$ and $J =10000$. The results are given in \Cref{supercond_n500}.
Overall the results are unsatisfactory:  While stacking did put a lot of
weight on XGB, the procedure advocated by \cite{Hamidieh:2018},
its predictive coverage is weak.  On the other hand, SVM, which performed better in
terms of coverage got
a low stacking weight.    This is likely due to the difference between coverage (what proportion of
new data points are in a PI) and minimzing $L^2$ predictive error.  
\begin{table}[ht]
\caption{\textbf{Small sample results for the Superconductivity data}:  The only
model with reasonable coverage, SVM, has a low stacking weight.  Also, the stacking average
while giving superb coverage, does so at the cost of high variance. -- larger than the variance
of any single model.  This is consistent with high between-models variance.}
\label{supercond_n500}
\centering
\begin{tabular}{rrrrrrr}
  \hline
 & STK avg & LM & NN & PPR & SVM & XGB \\
  \hline
Stacking weights &  & 0.10 & 0.26 & 0.12 & 0.01 & 0.51 \\
 Pred.  Variance & 398.08 & 237.33 & 260.46 & 57.06 & 172.11 & 69.39 \\
  Coverage & 1.00 & 0.87 & 0.79 & 0.26 & 1.00 & 0.79 \\
   \hline
\end{tabular}
\end{table}

Using only $n=499$, the stacking predictive variance decompositions is
\begin{align*}
Var(Y_{500})({\cal D}_{499}) & = E_{V}Var(Y_{500} \vert V)({\cal D}_{499}) + Var_{V}E(Y_{500} \vert V)({\cal D}_{499}) \\
& = 135.85 + 262.23 \\
& = 398.08.
\end{align*}
To test whether the between-models variance term matters, we have the hypotheses
$$
H_0:  E\left( \frac{Var_{V}E(Y_{500} \vert  V)({\cal D}_{499})}{Var(Y_{500} )({\cal D}_{499})}\right) \geq \tau
$$
versus
$$
H_1:  E\left( \frac{Var_{V}E(Y_{500} \vert  V)({\cal D}_{499})}{Var(Y_{500} )({\cal D}_{499})}\right) < \tau,
$$
and the test statistic $\bar{z} = \frac{262.23}{398.08} =  0.66$.  For $\tau = 0.05$ we obtain 
$\widehat{ASL}=1$ and cannot reject the null. 
In this case,  we cannot reject the null for any reasonable value of $\tau$.  This
confirms what \Cref{supercond_n500} showed, namely that the between-models 
variance is much bigger than the between-predictions
within-models variance.


For contrast we redo the analysis using all the available data. Here we let
$B=50$,  and $J=5000$.   Note that here we only used 50
bootstrap samples due to computational burden.  The results are given in 
\Cref{tab_supercond_all_data1}.    With the larger sample size
we find that all coverages are one and superficially if we had to choose
one method it would be XGB.
\begin{table}[H]
\caption{\textbf{Re-analyzing with all available data}:   The predictive variances in
this table are bigger than in \Cref{supercond_n500} but the overall
stacking variance is less than half of the earlier value.  This suggests the between 
models variance is less important than with $n=500$.}
\label{tab_supercond_all_data1}
\centering
\begin{tabular}{rrrrrrr}
  \hline
 & STK avg & LM & NN & PPR & SVM & XGB \\
  \hline
Stacking weights &  & 0.01 & 0.26 & 0.21 & 0.01 & 0.52 \\
Pred.   Variance & 173.73 & 308.60 & 315.28 & 184.14 & 155.32 & 78.71 \\
  Coverage & 1.00 & 1.00 & 1.00 & 1.00 & 1.00 & 1.00 \\
   \hline
\end{tabular}
\end{table}

Now the variance decompositions is
\begin{align*}
Var(Y_{21263})({\cal D}_{21262}) & = E_{V}Var(Y_{21263} \vert V)({\cal D}_{21262}) + Var_{V}E(Y_{21263} \vert V)({\cal D}_{21262}) \\
& = 166.57 + 7.16 \\
& = 173.73.
\end{align*}
Again, we wish to test if the between models term is a substantial portion of the total predictive variance. The hypotheses are
$$
H_0:  E\left( \frac{Var_{V}E(Y_{21263} \vert  V)({\cal D}_{21262})}{Var(Y_{21263} )({\cal D}_{21262})}\right) \geq \tau
$$
versus
$$
H_1:  E\left( \frac{Var_{V}E(Y_{21263} \vert  V)({\cal D}_{21262})}{Var(Y_{21263} )({\cal D}_{21262}))}\right) < \tau,
$$
and the test statistic is  $\bar{z}  = \frac{7.16}{173.73} = 0.041$. We used different choices of $\tau$ and observed the results in \Cref{tab_supercond_all_data2}.
It is seen that for $\tau=0.05$ there is not enough evidence to say $T$ is statistically less than 
$\tau$, but for $\tau \geq 0.06$ the test rejects the null.   That is,
the relative contribution of the between-models variance to the
total stacking predictive variance is roughly between five and six percent.
We suggest that if a larger value of $B$ could have been used, the threshold for rejecting the
null would likely decrease to around $\tau = .05$.

\begin{table}[ht]
\caption{\textbf{$\widehat{ASL}$ for different choices of $\tau$}:  The reliability of the entries
is potentially limited because $B$ is low.}
\label{tab_supercond_all_data2}
\centering
\begin{tabular}{rrrrrrr}
  \hline
 $\tau$& 0.05 & 0.06 & 0.07 & 0.08 & 0.09 & 0.10 \\
  \hline
  $\widehat{ASL}$ & 0.16 & 0.03 & 0.003 & 0.0003 &  0 &  0\\
   \hline
\end{tabular}
\end{table}
Thus, with $n=500$, we could not reject the null at any reasonable value of $\tau$
however with the full data set
we could reject the null at $\tau$ around 6\%.

As a final point for this section, we confirm that by looking at \Cref{tab_supercond_all_data1}
and the testing results, the preferred method of \cite{Hamidieh:2018} is well-justified.
It gives high coverage and the smallest variance among the alternatives we used.
Moreover, XGB received the highest stacking weight, presumably because it had
the smallest cross-validated error.

\section{Discussion}
\label{discussion}

Here we have proposed a decomposition of the posterior predictive variance and the
stacking predictive variance.  The decompositions are based on representing modeling
choices by a discrete random variable $V$ and then iterating the
law of total variance for each component of $V$.  The predictive variances control the
width of prediction intervals so our decomposition lets us assess the contribution of each
source of variability in $V$ to the overall variance.  We proposed a testing procdure to
assess the relative contributions of the terms in the decomposition so that we
can, in principle, eliminate some components of $V$ thereby simplifying the
resulting prediction intervals where possible.  We show how our analysis proceeds
in a series of examples, three of which are extensions of an earlier analysis
presented in \cite{Draper:1995}.  We verify that our methods give intuitively
plausible results for multiple choices of $V$.

Our analysis has analogies to the classical Cochran's Theorem decomposition of total
squared error into a sum of quadratic forms with independent $\chi^2$ distributions.
We do not find as neat a distributional form, however, we show that the terms in
our decomposition of the total predictive variance
correspond to sums of $\chi^2$ random variables.  

A recurrent theme in our findings is the discrepancy between the relative contribution
of a variance term to the total variance and its absolute level.  
The relative importance of a term depends
on the sample size differently from the total variance.  In particular,
we see that if the
absolute level of variance is small enough, then it is not important how much
each term in the decomposition contributes.  

We conclude with the observation that there may be two different choices of 
$V=V_K$ that an analyst
may want to consider.   This leads to the question as to how to choose 
one over the other.   In Sec. \ref{example} and in Sec. \ref{superconduct}
we faced a special case of this problem when we reduced a one dimensional $V$ 
to a single model.  Our approach can be formalized as the following empirical
optimization.  Recall that
in expanding our model list, we want to ensure we have close to the proper coverage 
and the smallest variance possible among model lists with good coverage. This leads 
us to choose $K$ and $V$ in the following manner. 
First calculate estimated coverage using $g$-fold cross-validation or 
$g$ bootstrapping samples and define the 
estimated coverage to be 
$$
\hat{C}({\cal{V}}^{(K)}) = \frac{1}{g}\sum^g_{i=1}I_{ \{y_{i, new} \in PI({\cal{V}}^{(k)}) \}},
$$
where ${\cal{V}}^K$ is the model lsit corresponding to $V$.
Then for given $\alpha, \delta >0$, we choose 
$$
\hat{K} = \arg  \min_{k \in \{k \vert \hat{C}({\cal{V}}^{(k)})\in (1-\alpha -\delta, 1 - \alpha  + \delta) \}} Var(Y_{n+1} \vert {\cal{D}}_n)({\cal{V}}^{(k)}).
$$
That is, we choose the value of $K$ and the corresponding $V$
to minimize the posterior variance among all model lists that have 
estimated coverage $\delta$-close to the nominal $1-\alpha$ coverage. 
Despite this data-driven proposal the problem of model list selection remains 
both difficult and open.

\appendix

\section*{Acknowledgments}
Dustin acknowledges funding from the University of Nebraska
Program of Excellence in Computational Science.  Both authors acknowledge computational support from the Holland Computing Center at the University of Nebraska.

\bibliographystyle{siamplain}
\bibliography{references}

\end{document}